\documentclass{article}
\usepackage[a4paper,left=2cm,right=2cm, top=3cm, bottom=3cm]{geometry}
\usepackage{graphicx,psfrag}
\usepackage{amsmath,amsthm,amssymb,amsfonts,amstext}
\usepackage{bbm}
\usepackage[english]{babel}
\usepackage[latin1]{inputenc}
\usepackage{lineno} 
\usepackage{enumitem,mdwlist} 
\usepackage{cases} 
\usepackage[colorlinks=true,linkcolor=blue,citecolor=blue]{hyperref} 
\usepackage{varwidth} 
\usepackage{subfigure} 
\usepackage{mathrsfs} 
\usepackage{caption} 

\newtheorem{theorem}{Theorem}[section]
\newtheorem{problem}[theorem]{Problem}

\newtheorem{corollary}[theorem]{Corollary}
\newtheorem{lemma}[theorem]{Lemma}

\makeatletter
\let\@fnsymbol\@arabic
\makeatother

\begin{document}
 
\title{On the Geodetic Hull Number of Complementary Prisms}

\author{
Erika M. M. Coelho$^{a,}$\thanks{E-mail: erikamorais@inf.ufg.br}, 
Hebert Coelho$^{a,}$\thanks{E-mail: hebert@inf.ufg.br}, \\
Julliano R. Nascimento$^{a,}$\thanks{E-mail: jullianonascimento@inf.ufg.br}, 
Jayme L. Szwarcfiter$^{b,c,}$\thanks{E-mail: jayme@nce.ufrj.br}}
\date{}
\maketitle

\begin{center}\vskip-17pt
$^a$ INF, Universidade Federal de Goi\'{a}s, GO, Brazil \\
$^b$ IM, COPPE, and NCE, UFRJ, RJ, Brazil \\
$^c$ IME, Universidade do Estado do Rio de Janeiro, RJ, Brazil \\
\end{center}

\hrule\vskip12pt

\begin{abstract}
Let $G$ be a finite, simple, and undirected graph and let $S$ be a set of vertices of $G$. In the geodetic convexity, a set of vertices $S$ of a graph $G$ is \textit{convex} if all vertices belonging to any shortest path between two vertices of $S$ lie in $S$. The \textit{convex hull} $H(S)$ of $S$ is the smallest convex set containing $S$. If $H(S) = V(G)$, then $S$ is a \textit{hull set}. The cardinality $h(G)$ of a minimum hull set of $G$ is the \textit{hull number} of $G$. The complementary prism $G\overline{G}$ of a graph $G$ arises from the disjoint union of the graph $G$ and $\overline{G}$ by adding the edges of a perfect matching between the corresponding vertices of $G$ and $\overline{G}$. Motivated by the work of Duarte et al. \cite{duarte2015complexity}, we determine and present lower and upper bounds on the hull number of complementary prisms of trees, disconnected graphs and cographs. We also show that the hull number on complementary prisms cannot be limited in the geodetic convexity, unlike the $P_3$-convexity.
\end{abstract}

\noindent {\small {\bf Keywords:} Geodetic convexity, Hull number, Complementary prisms}

\vskip12pt\hrule

\section{Introduction} 
\label{sec:introduction}

In this paper we consider finite, simple, and undirected graphs, and we use standard terminology in Graph Theory. For a finite and simple graph $G$ with vertex set $V(G)$, a \textit{graph convexity} on $V(G)$ is a collection $\mathcal{C}$ of subsets of $V(G)$ such that $\emptyset, V(G) \in \mathcal{C}$ and $\mathcal{C}$ is closed under intersections.
The sets in $\mathcal{C}$ are called \textit{convex sets} and the \textit{convex hull} $H_{\mathcal{C}}(S)$ in $\mathcal{C}$, of a set $S$ of vertices of $G$, is the smallest set in $\mathcal{C}$ containing $S$. 

Through the concepts of graph convexity we can model, and eventually solve, problems in contexts that requires some disseminating process, such as contamination \cite{balogh1998random, bollobas2006art, dreyer2009irreversible}, marketing strategies \cite{domingos2001mining, kempe2003maximizing, kempe2005influential}, spread of opinion \cite{brunetti2012minimum, dreyer2009irreversible}, spread of influence \cite{kempe2003maximizing, khoshkhah2014dynamic} and distributed computing \cite{flocchini2004dynamic, hassin2001distributed, mustafa2004listen, peleg2002local}.

Some natural convexities in graphs are defined by a set $\mathcal{P}$ of paths in $G$, in a way that a set $S$ of vertices of $G$ is convex if and only if for every path $P = v_0v_1 \dots v_l$ in $\mathcal{P}$ such that $v_0$ and $v_l$ belong to $S$, all vertices of $P$ belong to $S$. 

If we define $\mathcal{P}$ as the set of all shortest paths in $G$, we have the well-known \textit{geodetic convexity} \cite{caceres2006geodetic, dourado2016near, dourado2016pkfree, dourado2010hull, dourado2010some,  dourado2013caratheodory, everett1985hull, harary1981convexity}. 
The \textit{monophonic convexity} is defined by considering as $\mathcal{P}$ the set of all induced paths of $G$ \cite{costa2015inapproximability, dourado2010complexity, duchet1988convex, farber1986convexity} . The set of all paths of $G$ leads to the \textit{all paths convexity} \cite{changat2001all}. When $\mathcal{P}$ is the set of all triangle paths in $G$, then $\mathcal{C}$ is the \textit{triangle path convexity} \cite{changat1999triangle, dourado2016complexity}. 
The $P_3$\textit{-convexity}  is defined by considering $\mathcal{P}$ the set of all paths of $G$ with three vertices \cite{barbosa2012caratheodory, campos2015graphs, centeno2013geodetic, coelho2014caratheodory, dourado2013caratheodory, duarte2015complexity, nascimento2016complexity, penso2015complexity}.

Our work considers $\mathcal{C}$ the geodetic convexity. Given a graph $G$, the \textit{closed interval} $I[u,v]$ of a pair $u,v \in V(G)$ consists of $u, v$, and all vertices lying in any shortest path between $u$ and $v$ in $G$. For a set $S \subseteq V(G)$, the \textit{closed interval} $I[S]$ is the union of all sets $I[u,v]$ for $u,v \in S$. If $I[S] = S$, then $S$ is a \textit{convex set}. 
 
The \textit{convex hull} $H_{\mathcal{C}}(S)$ of $S$ is the smallest convex set containing $S$. Since a graph $G$ uniquely determines its convexity $\mathcal{C}$, we may write $H(S)$, instead of $H_{\mathcal{C}}(S)$. 
 The convex hull $H(S)$ can be formed from the sequence $I^p[S]$, where $p$ is a nonnegative integer, $I^0[S] = S$, $I^1[S] = I[S]$, and $I^p[S] = I[I^{p-1}[S]]$, for $p \geq 2$. When, for some $p$, we have $I^q[S] = I^p[S]$, for all $q \geq p$, then $I^p[S]$ is a convex set.  
Let $S$ and $X$ be sets of vertices of a graph $G$. If $X \subseteq H(S)$, we say that $X$ is \textit{contaminated} by the vertices of $S$, or $S$ \textit{contaminates} $X$.
If $H(S) = V(G)$ we say that $S$ is a \textit{hull set} of $G$. 
The cardinality $h(G)$ of a minimum hull set of $G$ is called the \textit{geodetic hull number}  of $G$, or simply \textit{hull number} of $G$. 

Everett and Seidman \cite{everett1985hull} introduced the concept of \textit{hull number} considering the geodetic convexity. For some later results see e.g. \cite{buckley1990distance, caceres2010geodetic, canoy2005hull,  canoy2006convexity, chartrand2000hull, dourado2010hull, hernando2005steiner}.  
Related to complexity aspects, Dourado et al. \cite{dourado2009computation} proved that, given a graph $G$ and an integer $k$, to decide whether the hull number of $G$ is at most $k$ is a problem NP-complete. Araujo et al. \cite{araujo2011hull} showed that the same problem is NP-complete even in bipartite graphs. Among other results, Dourado, Penso and Rautenbach \cite{dourado2016pkfree} proved that the hull number is NP-hard in $P_9$-free graphs.


It is worth mentioning that the studies of hull number were also extended for other graph convexities. For example, for general graphs, the hull number can be determined in polynomial time in triangle path convexity \cite{dourado2016complexity} and monophonic convexity \cite{dourado2010complexity}.
In $P_3$-convexity, the problem for general graphs is NP-complete \cite{centeno2011irreversible}. However, Duarte et al. \cite{duarte2015complexity} showed that the hull number of complementary prisms can be determined efficiently.

The \textit{complementary prisms} were introduced by Haynes et al. \cite{haynes2007complementary} as a variation of the well-known \textit{prism} of a graph \cite{hammack2011handbook}. For a graph $G$ with vertex set $V(G) = \{v_1, \dots , v_n\}$  and edge set $E(G)$, the \textit{complementary prism} of $G$ is the graph denoted by $G\overline{G}$ with vertex set $V(G\overline{G}) = \{v_1, \dots , v_n\} \cup \{\overline{v}_1, \dots , \overline{v}_n\}$ and edge set $$E(G\overline{G}) = E(G) \cup \{\overline{v}_i\overline{v}_j : 1 \leq i < j \leq n \mbox{ and } v_iv_j \notin E(G)\} \cup \{v_1\overline{v}_1, \dots , v_n\overline{v}_n\}.$$

Let $G$ be a graph and $\overline{G}$ its complement. For every vertex $v \in V(G)$ we denote $\overline{v} \in V(\overline{G})$ as its \textit{corresponding vertex}. 
In other words, the complementary prism $G\overline{G}$ of $G$ arises from de disjoint union of the graph $G$ and its complement $\overline{G}$ by adding the edges of a perfect matching joining corresponding vertices of $G$ and $\overline{G}$.

Duarte \cite{duarte2015prismas} determined the geodetic hull number for complementary prisms $G\overline{G}$ when $G$ is a path, a cycle, or a complete graph, and proved that the $P_3$-hull number of complementary prisms $G\overline{G}$ is limited when $G$ and $\overline{G}$ are connected. We extend these results by determining and presenting lower and upper bounds on the hull number for complementary prisms $G\overline{G}$ when $G$ is a tree, a disconnected graph or a cograph. We also prove that the geodetic hull number on the complementary prism $G\overline{G}$ cannot be limited when $G$ and $\overline{G}$ are connected graphs, unlike what happens in $P_3$-convexity.

This paper is divided in more three sections. In Section \ref{sec:preliminaries} we define the fundamental concepts and terminology. In Section \ref{sec:results} we present our results. We close with the conclusions and future works in Section \ref{sec:conclusions}.

\section{Preliminaries}
\label{sec:preliminaries}

Let $G$ be a graph. Given a vertex $v \in V(G)$, the \textit{open neighborhood} of $v$ is the set of neighbors of $v$, denoted by $N_G(v)$. 
The \textit{closed neighborhood} of $v$, denoted by $N_G[v]$, is the set $N_G[v] = N_G(v) \cup \{v\}$. The \textit{open neighborhood} of a set $U \subseteq V(G)$, denoted by $N_G(U)$, is the set of vertices of $V(G) \setminus U$ which are adjacent to some vertex of $U$. 

 A \textit{clique} of a graph $G$ is a subset of pairwise adjacent vertices in $G$. We say that $v$ is a simplicial vertex of $G$ if $N_G[v]$ induces a clique. According to Everett and Seidman \cite{everett1985hull}, every hull set of a graph $G$ contains the set of all simplicial vertices of $G$, as stated in Lemma \ref{lema:simpliciais}.

\begin{lemma}[Everett and Seidman \cite{everett1985hull}]
\label{lema:simpliciais}
For every hull set $S$ of a graph $G$, $S$ contains the set of all simplicial vertices of $G$.
\end{lemma}

Given two vertices $u$ and $v$ of a graph $G$, the \textit{distance} $d_G(u,v)$ is the length of a path linking $u$ and $v$ with minimum number of edges. We say that $v$ is \textit{reachable} from $u$, if there exists a path $P$ from $u$ to $v$ in $G$. Let $P = u_0u_1\dots u_k$ a path in $G$. For every vertex $u_i$, for $1 \leq i \leq k$, we say that the vertex $u_{i-1}$ is a \textit{predecessor} of $u_i$  in $P$.

A graph $G$ is called \textit{connected} if any two of its vertices are linked by a path in $G$. Otherwise, $G$ is called \textit{disconnected}. A maximal connected subgraph of $G$ is called a \textit{connected component} or \textit{component} of $G$.
A component $G_i$ of a graph $G$ is \textit{trivial}, if  $|V(G_i)| = 1$, and \textit{nontrivial} otherwise.
The greatest distance between any two vertices in $G$ is the \textit{diameter} of $G$, denoted by $diam(G)$.

As described in Introduction, in \cite{duarte2015prismas}, Duarte shows results on the geodetic hull number for complementary prisms, that follow below.

\begin{theorem}[Duarte \cite{duarte2015prismas}]\label{theo:duarte}
Let $G$ be a graph.
\begin{enumerate}
\item If $G = K_n$, then $h(G\overline{G}) = n$, for $n \geq 2$;
\item If $G = P_n$, then $h(G\overline{G})=
\begin{cases}
3,&\mbox{if}\quad n = 3,\\
2, &\mbox{otherwise}.
\end{cases}$
\item If $G = C_n$, then $h(G\overline{G})=
\begin{cases}
2,&\mbox{if}\quad n \geq 6,\\
3, &\mbox{otherwise}.
\end{cases}$
\end{enumerate}
\end{theorem}

\section{Results} 
\label{sec:results}

This section is intended to present our results. 
We begin by showing in Lemma \ref{lema:simpliciaisPrisma} a property related to simplicial vertices in a graph $G$ and in its complement $\overline{G}$.

\begin{lemma}
\label{lema:simpliciaisPrisma}
Let $G$ be a graph. If $u$ is a simplicial vertex in $G$ and $\overline{u}$ is a simplicial vertex in $\overline{G}$, then every hull set $S$ of $G\overline{G}$ intersects $\{u,\overline{u}\}$.
\end{lemma}

\begin{proof}
Let $G$ be a graph. Let $u$ be a simplicial vertex in $G$ and $\overline{u}$ be a simplicial vertex in $\overline{G}$. Suppose, by contradiction, that there exists a hull set $S$ of $G\overline{G}$ such that $S \cap \{u,\overline{u}\} = \emptyset$.

Since $S$ is a hull set of $G\overline{G}$, we have that $u,\overline{u} \in H(S)$.
Since the distance between any vertex of $G$ to any vertex of $\overline{G}$ is at most two, any path containing $\{u,\overline{u}\}$ has distance at least $3$. Then, $u$ and $\overline{u}$ were not contaminated simultaneously. 
Furthermore, we have that $I[x,\overline{y}]$, with $x \in V(G) \setminus \{u\}$ and $\overline{y} \in V(\overline{G}) \setminus \{\overline{u}\}$, does not contain $u$ nor $\overline{u}$. So, the vertex $u$ or $\overline{u}$ were not contaminated by $x$ and $\overline{y}$. Thus, the vertex $u$ ($\overline{u}$) was contaminated by a vertex of $G$ ($\overline{G}$).

\bigskip
\noindent {\bf Case 1:} {\it The vertex $u$ was contaminated by $x,y \in V(G) \setminus \{u\}$.}
\bigskip

\noindent Consider that $x,y \in V(G) \setminus \{u\}$. Since $u \in I[x,y]$, then $xy \notin E(G)$. But since $u$ is simplicial in $G$, we have that $N_G[u]$ is a clique, therefore any shortest path joining $x$ and $y$ does not contain $u$. Thus, we have that $u \notin H(S)$, a contradiction.

\bigskip
\noindent {\bf Case 2:} {\it The vertex $\overline{u}$ was contaminated by $\overline{x},\overline{y} \in V(\overline{G}) \setminus \{\overline{u}\}$.}
\bigskip

\noindent Consider that $\overline{x},\overline{y} \in V(\overline{G}) \setminus \{\overline{u}\}$. Analogously to Case 1, since $\overline{u} \in I[\overline{x},\overline{y}]$, then $\overline{x}\overline{y} \notin E(\overline{G})$. But since $\overline{u}$ is simplicial $\overline{G}$, we have that $N_{\overline{G}}[\overline{u}]$ is a clique, then any shortest path linking $\overline{x}$ and $\overline{y}$ does not contain $\overline{u}$, therefore $\overline{u} \notin H(S)$, a contradiction.

\bigskip
So, $S$ is a hull set of  $G\overline{G}$, then $S \cap \{u,\overline{u}\} \neq \emptyset$.
\end{proof}

\subsection{Trees} 
\label{subsec:trees}

Now, we proceed with the results of the hull number for the complementary prism $G\overline{G}$ when $G$ is a tree. A special case of tree, called \textit{star} $S_n$ is the complete bipartite graph $K_{1,n}$.

\begin{theorem}
\label{theo:arvoresGeodesica}
Let $T$ be a tree and consider an integer $n \geq 3$. Then:  
$$h(T\overline{T})=
\begin{cases}
n+1,&\mbox{if}\quad T = S_n,\\
2, &\mbox{otherwise}.
\end{cases}$$
\end{theorem}

\begin{proof}
Consider the cases: $T$ is a star or not.

First, if $T$ is a star, let  $V(T)=\{u_0,u_1,u_2,\dots,u_n\}$ and $E(T) = \{u_0u_i : 1 \leq i \leq n \}$, where $n \geq 3$. Let $S$ be a hull set of $T\overline{T}$.

By the definition of $T$, we have that $u_0 \in T$ is adjacent to $u_i$, for every $1 \leq i \leq n$, then $\overline{u}_0 \in \overline{T}$ has no adjacent vertices in $\overline{T}$. This implies that $\overline{u}_0$ is a simplicial vertex in $V(G\overline{G})$. Then, by Lemma \ref{lema:simpliciais}, $S$ contains $\overline{u}_0$, therefore $h(T\overline{T}) \geq 1$.

Still by the definition of $T$, the vertex $u_i$ is simplicial in $T$, for every $1 \leq i \leq n$, and $\overline{u}_i$ is simplicial in $\overline{T}$, for every $1 \leq i \leq n$. This way, Lemma  \ref{lema:simpliciaisPrisma}  implies that $S$ contains at least one vertex from each set $U_i = \{u_i,\overline{u}_i\}$, for every $1 \leq i \leq n$, then $h(T\overline{T}) \geq n+1$.

For the upper bound, consider $S = \{\overline{u}_0,\overline{u}_1,\overline{u}_2,\dots,\overline{u}_n\}$. We have that $d_{T\overline{T}}(\overline{u}_0,\overline{u}_i) = 3$, for every $1 \leq i \leq n$. Since $u_0$ and $u_i$ belong to a shortest path from $\overline{u}_0$ to $\overline{u}_i$, then $u_0,u_i \in I[S]$. Thus, $V(T\overline{T}) = I[S] = H(S)$, therefore $S$ is a hull set of $T\overline{T}$. Since $|S| = n+1$, then $h(T\overline{T}) = n+1$, when $T$ is a star.

For second equality, since $h(P_n\overline{P}_n) = 2$, for $n \neq 3$ (Theorem \ref{theo:duarte}), and $P_3$ is a star, suppose that $|V(T)| \geq 5$.

Let $T \neq P_n$. Since $T$ is not a star and $|V(T)| \geq 5$, there exist vertices $u,v \in T$ such that $d_T(u,v) \geq 3$. Consider $S = \{u,v\}$ such that $u,v \in V(T)$ and the distance between $u$ and $v$ in $T$ is equal to $3$. 
Since $d_T(u,v) = 3$, there exist $x,y \in V(T)$ such that $x$ and $y$ belong to a shortest path joining $u$ and $v$, let $P_1 = u, x, y, v$. The vertices $\overline{u},\overline{v} \in V(\overline{T})$ also belong to a shortest path joining $u$ and $v$, let $P_2 = u, \overline{u}, \overline{v}, v$. Then $I[S] \supseteq \{u,v,x,y,\overline{u},\overline{v}\}$.
Since $d_{T\overline{T}}(x,\overline{v}) = 2$ and $d_{T\overline{T}}(y,\overline{u}) = 2$, we have that $I^2[S] \supseteq \{u,v,x,y,\overline{u},\overline{v},\overline{x},\overline{y}\}$.

Since $ux \in E(T)$, $\overline{u}\overline{x} \notin E(\overline{T})$, then $d_{T\overline{T}}(\overline{u},\overline{x}) \geq 2$. 
Let $W = \{w \in V(T) : w \notin N_{T}(\{u,x\})\setminus\{y\} \}$ and $\overline{W} = \{ \overline{w} \in V(\overline{T}) : \overline{w} \mbox{ is corresponding to } w \in W \}$.
For every vertex $w \in W$, we have that  every $\overline{w} \in \overline{W}$ is adjacent to $\overline{u}$ and $\overline{x}$, thus $\overline{w}$ belongs to a shortest path between $\overline{u}$ and $\overline{x}$, consequently $\overline{W} \subseteq I^3[S]$. 
Since $yv \in E(T)$, $\overline{y}\overline{v} \notin E(T)$, then $d_{T\overline{T}}(\overline{y},\overline{v}) \geq 2$. 
Let $Z = N_T(\{u,x\})\setminus\{y\}$ and $\overline{Z} = \{ \overline{z} \in V(\overline{T}) : \overline{z} \mbox{ is corresponding to } z \in Z \}$.
For every vertex $z \in Z$,  $z$ is not adjacent to $y,v$, then $\overline{z} \in \overline{Z}$ is adjacent to $\overline{y}$ and $\overline{v}$. Thus $\overline{z}$ belongs to a shortest path between $\overline{y}$ and $\overline{v}$ which implies that $\overline{Z} \subseteq I^3[S]$.  Thus, $V(\overline{T}) \subseteq H(S)$.

Since $T$ is connected, every vertex $t \in V(T)\setminus \{u,x,y,v\}$ is reachable from a vertex in $\{u,x,y,v\}$ and since $V(\overline{T}) \subseteq I^3[S]$, then $t \in I^\alpha[S]$ for $4 \leq \alpha \leq h$, in which $h$ is the height of the breadth search tree of $T$, starting from $\{u,x,y,v\}$.
Since $H(S) = I^\alpha[S] = V(T\overline{T})$, then $S$ is a hull set of $T\overline{T}$, therefore $h(T\overline{T}) = 2$.

\end{proof}

Figure \ref{subfig:prisma_estrela} contains a complementary prism $S_n\overline{S}_n$ and Figure \ref{subfig:prisma_arvore} contains a complementary prism  $T\overline{T}$ with $T \neq S_n$.  The black vertices represent a hull set of each complementary prism.

\begin{figure}[ht]
\centering
{\setlength{\fboxsep}{9pt} 
\setlength{\fboxrule}{0.3pt} 
\fbox{
\subfigure[Complementary prism of a star $S_n$.]
{
\psfrag{Sn}{$S_n$} \psfrag{Sbn}{$\overline{S}_n$}
\psfrag{u0}{$u_0$} \psfrag{ub0}{$\overline{u}_0$} 
\psfrag{u1}{$u_1$} \psfrag{ub1}{$\overline{u}_1$} 
\psfrag{u2}{$u_2$} \psfrag{ub2}{$\overline{u}_2$} 
\psfrag{un}{$u_n$} \psfrag{ubn}{$\overline{u}_n$} 
\psfrag{...}{...}
\includegraphics[width=0.32\textwidth]{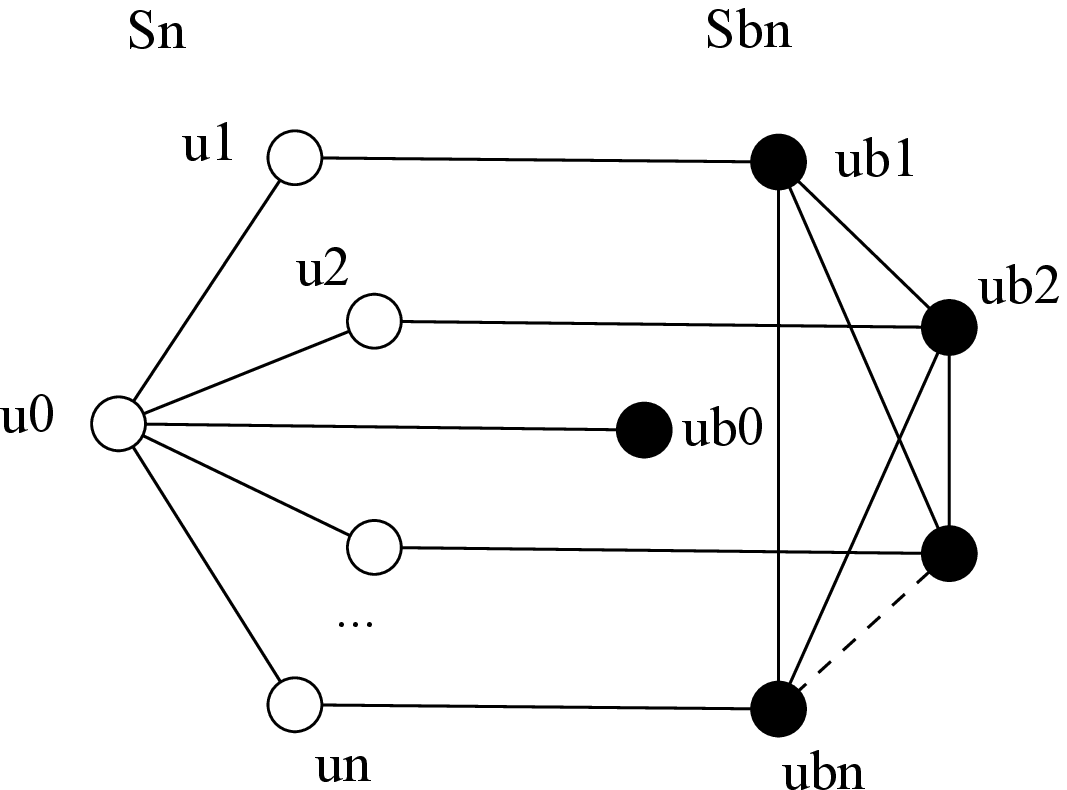}
\label{subfig:prisma_estrela}
} \qquad
\subfigure[Complementary prism of a tree $T$.]
{
\psfrag{T}{$T$} \psfrag{Tb}{$\overline{T}$}
\psfrag{u}{$u$} \psfrag{ub}{$\overline{u}$} 
\psfrag{x}{$x$} \psfrag{xb}{$\overline{x}$} 
\psfrag{y}{$y$} \psfrag{yb}{$\overline{y}$} 
\psfrag{v}{$v$} \psfrag{vb}{$\overline{v}$} 
\psfrag{z}{$z$} \psfrag{zb}{$\overline{z}$} 
\psfrag{w}{$w$} \psfrag{wb}{$\overline{w}$} 
\includegraphics[width=0.28\textwidth]{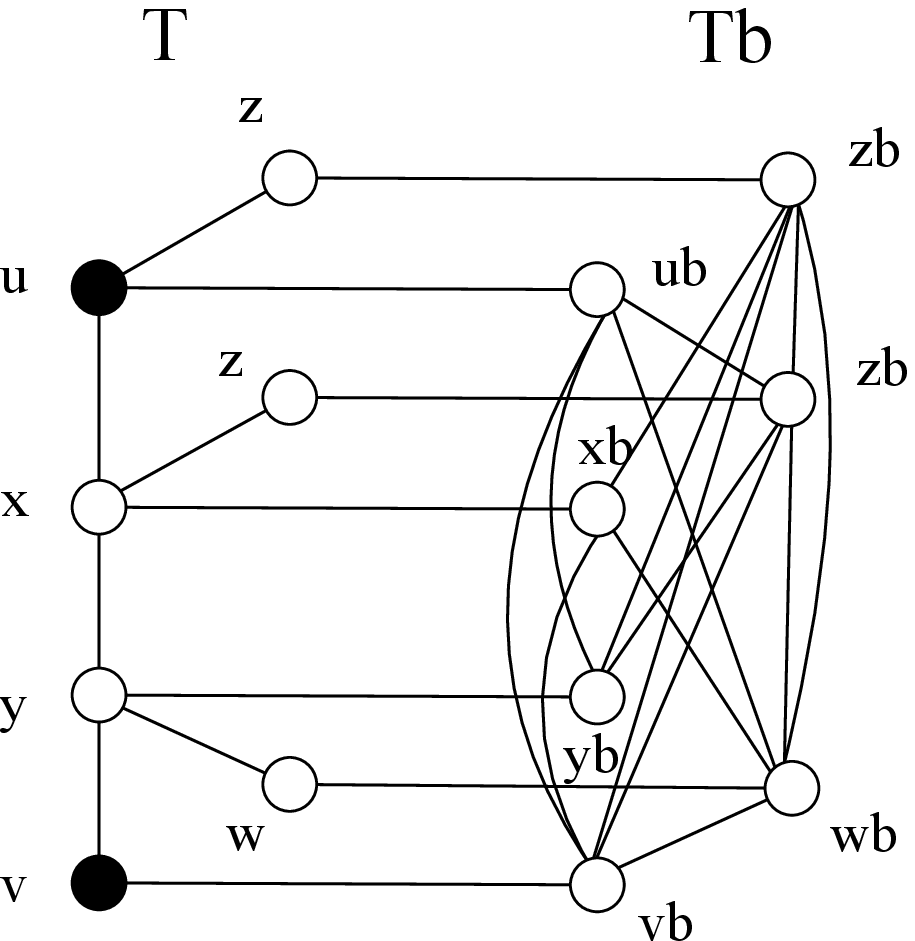}
\label{subfig:prisma_arvore}
}}
} 
\caption{Trees of the cases of Theorem \ref{theo:arvoresGeodesica}.}
\label{fig:prismas_arvores}
\end{figure}

\subsection{Disconnected Graphs}

In $P_3$-convexity, Duarte et al. \cite{duarte2015complexity} study the hull number in complementary prisms $G\overline{G}$ when $G$ is a disconnected graph. We show in Theorem \ref{theo:Gdesconexo} a similar result, considering the geodetic convexity.

\begin{theorem}
\label{theo:Gdesconexo}
Let $G$ be a disconnected graph. If $G$ has $k \geq 2$ connected components, in which at least two of them are nontrivial, then $h(G\overline{G}) = k + 1$.
\end{theorem}

\begin{proof}
Let $G$ be a disconnected graph with $V(G) = \{u_1,\dots,u_n\}$. Suppose that $G$ has at least two nontrivial connected components. Let $G_1, \dots, G_k$ be the connected components of $G$.  Let $\overline{G}_1, \dots, \overline{G}_k$ the subgraphs induced by $\overline{u}_i \in V(\overline{G})$ corresponding to the vertices  $u_i$ from each connected component of $G$.

Since $G$ is disconnected, every component $G_{i}$, $1\leq i \leq k$, is not contaminated by vertices from other components $G_j$, $1\leq j \leq k$, $i \neq j$. Since every vertex of $\overline{G}_{i}$, $1\leq i \leq k$, is adjacent to every vertex of $\overline{G}_{j}$, $1\leq j \leq k$, for $i \neq j$, then the distance between any two vertices of $\overline{G}$ is at most $2$. Thus, every component $G_{i}$, $1\leq i \leq k$, is not contaminated only by vertices of $\overline{G}$, which implies that $h(G\overline{G}) \geq k$. Since every set of vertices of $G\overline{G}$ containing only one vertex $u_i$ of each component $G_{i}$, $1\leq i \leq k$, contaminates only corresponding vertices $\overline{u}_i$ in  $\overline{G}_i$, then $h(G\overline{G}) \geq k+1$. 

Let $G_r$ and $G_{s}$ two nontrivial connected components of $G$. Let $u_1,u_2 \in V(G_r)$ such that $u_1u_2 \in E(G)$. Let $S = \{u_1,\overline{u}_2\} \cup \{u_i \in G_i : 1 \leq i \leq k, i \neq r\}$. 

Since $d_{G\overline{G}}(u_1,\overline{u}_2) = 2$, $u_1u_2\overline{u}_2$ is a shortest path between $u_1$ and $\overline{u}_2$, then $u_2 \in I[S]$. We also have that $u_1\overline{u}_1\overline{u}u$ is a shortest path between $u_1$ and $u$, for every $u \in S$ from each $G_i$, $1 \leq i \leq k$ with $i \neq r$, then $\overline{u}_1,\overline{u} \in I[S]$.

Since $u_1u_2 \in E(G)$, then $\overline{u}_1\overline{u}_2 \notin E(\overline{G})$. 
Since $G$ has connected components $G_1, \dots, G_k$, then every vertex of each induced subgraph  $\overline{G}_i$ of $\overline{G}$ is adjacent to every vertex of each induced subgraph $\overline{G}_{j}$ of $\overline{G}$, for $1 \leq i,j \leq k$, $i \neq j$.
Since  $d_{G\overline{G}}(\overline{u}_1,\overline{u}_2) = 2$, every vertex $\overline{x}$ in $\overline{G}_j$, $1 \leq j \leq k$ with $j \neq r$, is in a shortest path between  $\overline{u}_1$ and $\overline{u}_2$, thus $I^2[S]$ contains every vertex of $\overline{G}_j$, $1 \leq j \leq k$, $j \neq r$.

Since $G_{s}$ is a nontrivial connected component of $G$, then there exist at least two adjacent vertices in $G_s$, $u_3$ and $u_4$. 
Since $\overline{u}_3,\overline{u}_4 \in I^2[S]$, because $r \neq s$, and $d_{G\overline{G}}(\overline{u}_3,\overline{u}_4) = 2$, then every vertex $\overline{w} \in \overline{G}_r$ is in a shortest path linking $\overline{u}_3$ and $\overline{u}_4$, which implies that $I^3[S]$ contains every vertex $\overline{w} \in \overline{G}_r $. Thus, $I^3[S]$ contains $V(\overline{G})$.

Remains to show that $V(G) \subseteq H(S)$. Consider now, each connected component $G_i$, for $1 \leq i \leq k$. Since each component $G_i$ is connected, then there exists a path joining each pair of vertices in each $G_i$. Since there exists at least one vertex $u$ contaminated in each $G_i$, then every vertex $u_l \in G_i$ is in a shortest path between $\overline{u}_l$ and $u$, or between $\overline{u}_l$ and its predecessor previously contaminated, which implies that $u_l \in H(S)$. This way, $V(G) \subseteq H(S)$, then $V(G\overline{G}) = H(S)$. Therefore, $S$ is a hull set of $G\overline{G}$ and $h(G\overline{G}) = k + 1$.
\end{proof}

Ilustrating Theorem \ref{theo:Gdesconexo}, Figure \ref{fig:Gdesconexo} shows a complementary prism $G\overline{G}$ of a disconnected graph $G$ with at least two nontrivial connected components. Circles represent the connected components $G_i$, dashed circles represent the subgraphs  $\overline{G}_i$ and dashed lines represent the set of edges joining every vertex of $\overline{G}_i$ to every vertex of $\overline{G}_j$, $i \neq j$. Black vertices represent a hull set of  $G\overline{G}$.

\begin{figure}[htb]
\centering
\captionsetup{justification=centering}
{\setlength{\fboxsep}{9pt} 
\setlength{\fboxrule}{0.3pt} 
\fbox{
\psfrag{G}{$G$} \psfrag{Gb}{$\overline{G}$}
\psfrag{up}{$u_1$} \psfrag{ubp}{$\overline{u}_1$}
\psfrag{uq}{$u_2$} \psfrag{ubq}{$\overline{u}_2$}
\psfrag{v}{$u_4$} \psfrag{x}{$\overline{u}_4$}
\psfrag{Gi}{$G_r$} \psfrag{Gbi}{$\overline{G}_r$}
\psfrag{Gi'}{$G_{s}$} \psfrag{Gbi'}{$\overline{G}_{s}$}
\psfrag{vi'}{$u_{3}$} \psfrag{vbi'}{$\overline{u}_{3}$} 
\psfrag{Gj}{$G_j$} \psfrag{Gbj}{$\overline{G}_j$}
\psfrag{vj}{$u$} \psfrag{vbj}{$\overline{u}$} 
\psfrag{w}{$\overline{w}$}
\includegraphics[width=0.38\textwidth]{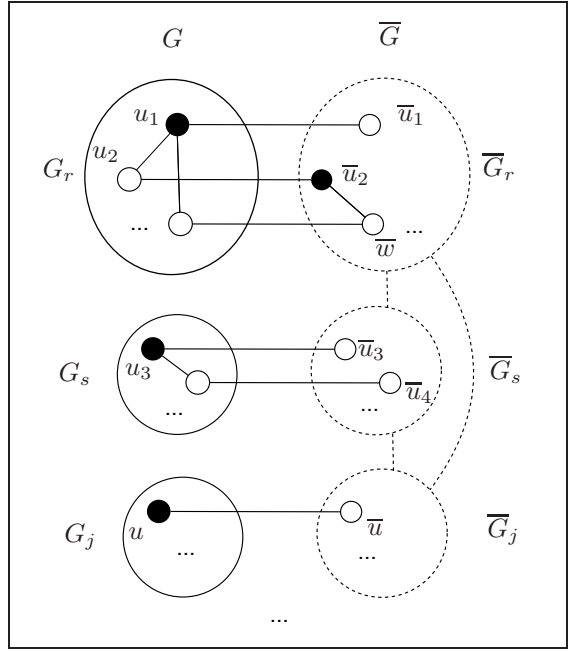}
}}
\caption{Complementary prism $G\overline{G}$ of a disconnected graph $G$, \\with nontrivial connected components $G_r$ and $G_s$.}
\label{fig:Gdesconexo}
\end{figure}

Now, let $G$ be a disconnected graph with exactly one nontrivial connected component. We show lower and upper bounds on the hull number of the complementary prism $G\overline{G}$, in Theorems \ref{theo:limiteInferiorGdesconexo} and \ref{theo:limiteSuperiorGdesconexo}, respectively. 

\begin{theorem}\label{theo:limiteInferiorGdesconexo}
Let $G$ be a disconnected graph. If $G$ has exactly one nontrivial connected component and $t > 0$ trivial components, then $h(G\overline{G}) \geq t + 2$.
\end{theorem}

\begin{proof}
Let $G$ be a disconnected graph, with exactly one nontrivial connected component, denoted by $G_1$. Consider $v_1, \dots, v_t$ the vertices of the $t > 0$ trivial components of $G$ and $V(G_1) = \{u_1, \dots, u_p\}$. Let $S \subseteq V(G\overline{G})$ a hull set of $G\overline{G}$.

Since $t > 0$ and every vertex $v_i$ of each trivial component of $G$, for $1 \leq i \leq t$, is simplicial in $G\overline{G}$, by Lemma \ref{lema:simpliciais}, $v_i \in S$. This implies that $h(G\overline{G}) \geq t$. Since the distance between any two vertices of $\overline{G}$ is at most $2$, we have that no shortest path between two vertices of  $\overline{G}$ contains vertices of $G_1$. Then, $S \cap V(G_1) \neq \emptyset$, which implies that $h(G\overline{G}) \geq t + 1$.
 Since $G$ has a nontrivial component $G_1$, then $|V(G_1)| \geq 2$. But $t + 1$ vertices of $G$, with only one them in $G_1$ contaminates only the corresponding vertices in $G\overline{G}$. Then $|S \cap V(G_1\overline{G}_1)| \geq 2$, and $h(G\overline{G}) \geq t+2$.
\end{proof}

\begin{theorem}\label{theo:limiteSuperiorGdesconexo}
Let $G$ be a disconnected graph with exactly one nontrivial connected component. Consider $G_1$ the nontrivial component of $G$, $\overline{G}_1$ the corresponding component of $G_1$ in $\overline{G}$ and $t > 0$ the number of trivial components in $G$. Then: 
\begin{numcases}{h(G\overline{G}) \leq }
h(G_1) + t, & \mbox{if $diam(G_1) \leq 3$}; \label{theo:limiteSuperiorGdesconexoMenorIgual3}\\
t+2, & \mbox{otherwise}. \label{theo:limiteSuperiorGdesconexoCasoContrario}
\end{numcases}
\end{theorem}

\begin{proof}
Let $G$ be a disconnected graph with exactly one nontrivial connected component, denoted by $G_1$. Consider $v_1, \dots, v_t$ the vertices of the $t > 0$ trivial components of $G$ and $V(G_1) = \{u_1, \dots, u_p\}$.

\eqref{theo:limiteSuperiorGdesconexoMenorIgual3} Let $S_1$ be a minimum hull set of $G_1$ and suppose that $diam(G_1) \leq 3$. Let $S \subseteq V(G\overline{G})$ such that $S = S_1 \cup \{v_1, \dots, v_t\}$.

Since $S_1$ is a hull set of $G_1$ and the maximum distance between each pair of vertices of $G_1$ is at most $3$, we have that the distances between each pair of vertices of $G_1$ are the same in $G\overline{G}$, then $H(S)$ contains $V(G_1)$. Consequently, $V(G) \subseteq H(S)$, since $v_i \in S$, for every $1 \leq i \leq t$.

Since $G$ is disconnected, then $d_{G\overline{G}}(u,v_i) = 3$, for every $u \in V(G_1)$, for every $1 \leq i \leq t$. Thus, we have that $\overline{u}, \overline{v}_i \in I[u,v_i]$, for every $u \in V(G_1)$ and $1 \leq i \leq t$, which implies that $V(\overline{G}) \subseteq H(S)$. Thus, $H(S) = V(G\overline{G})$, and $S$ is a hull set of $G\overline{G}$. Therefore $h(G\overline{G}) \leq h(G_1) + t$. 

\eqref{theo:limiteSuperiorGdesconexoCasoContrario} Let $S_1$ be a minimum hull set of $G_1$ and suppose that $diam(G_1) > 3$. 

Since $diam(G_1) > 3$, there exist at least two vertices $x, y \in V(G_1)$ such that $d_G(x,y) > 3$. Thus, we can define, without loss of generality, a path $P = u_1u_2u_3u_4u_5$ in $G$ such that $d_G(u_1,u_5) = 4$.

Let $S \subseteq V(G\overline{G})$ such that $S = \{u_1,u_4\} \cup \{v_i\}$, for every $1 \leq i \leq t$.

Since $d_{G\overline{G}}(u_1,u_4) = 3$, then $u_2,u_3 \in I[S]$. Still in the same iteration of the interval operation, $\overline{u}_1,\overline{u}_4, \overline{v}_i \in I[S]$, for every $1 \leq i \leq t$, since $u_1\overline{u}_1\overline{v}_iv_i$ is a shortest path from $u_1$ to $v_i$, for every $1 \leq i \leq t$. Then, since $u_2$ and $u_3$ are contaminated, we have that $\overline{u}_2, \overline{u}_3 \in I^2[S]$, since $u_2\overline{u}_2\overline{v}_iv_i$ and $u_3\overline{u}_3\overline{v}_iv_i$ are shortest paths, respectively, joining $u_2$ and $v_i$, and $u_3$ and $v_i$, for some $1 \leq i \leq t$.

To complete the proof that $S$ is a hull set of $G\overline{G}$, we analyse a vertex $z \in V(G)$ depending on the number of its neighbors in $P' = \{u_1,u_2,u_3,u_4\}$, to verify whether  $z$ or $\overline{z}$ can be contaminated.  Since $d_{G\overline{G}}(u_1,u_4) = 3$, we disregard the cases in which $z$ is adjacent to both $u_1$ and $u_4$.

\bigskip
\noindent {\bf Case 1:} {\it $|N_G(z)\cap P'| \leq 1$.}
\bigskip

\noindent  In this case, $\overline{z}$ is adjacent to two vertices of $\{\overline{u}_1,\overline{u}_2, \overline{u}_3, \overline{u}_4\}$, say $\overline{u}_i$, and $\overline{u}_j$. Then $\overline{z} \in I[\overline{u}_i, \overline{u}_j]$.

\bigskip
\noindent {\bf Case 2:} {\it $|N_G(z)\cap P'| = 2$.}
\bigskip

\noindent Suppose that $z$ has two neighbors $u_i,u_j \in P'$, for $i,j \in \{1,\dots,4\}$, $i \neq j$. If $i,j \in \{1,2\}$ ($i,j \in \{3,4\}$), $i \neq j$, we repeat the arguments in Case 1.
If $i,j \in \{2,3\}$, $i \neq j$, we have that $\overline{z} \in I[\overline{u}_4, \overline{u}_5]$.
If $i,j \in \{1,3\}$, ($i,j \in \{2,4\}$) $i \neq j$, we have that $z \in I[u_1,u_4]$. Since $z$ is contaminated, then $\overline{z} \in I[z,v]$, for $v$ of any trivial component of $G$.

\bigskip
\noindent {\bf Case 3:} {\it $|N_G(z)\cap P'| = 3$.}
\bigskip

\noindent Suppose that $z$ has three neighbors $u_i,u_j,u_k \in P'$, for $i,j,k \in \{1,\dots,4\}$, $i \neq j \neq k$. Since $d_{G\overline{G}}(u_1,u_4) = 3$, $i,j,k \in \{1,2,3\}$ or $i,j,k \in \{2,3,4\}$, $i \neq j \neq k$. Then we have that $z \in I[u_1,u_4]$. Since $z$ is contaminated, then, as in Case 2,  $\overline{z}$ is contaminated.

\bigskip

\noindent By the cases above, we conclude that $\overline{z} \in H(S)$. Thus, we have that $V(\overline{G}) \subseteq H(S)$. Since $G_1$ is connected and $S \cap V(G_1) \neq \emptyset$, every uncontaminated vertex $z \in V(G)$ is contaminated by the vertices from a shortest path between $\overline{z}$ and its contaminated predecessor in $G_1$, then $H(S) = V(G\overline{G})$. Therefore, $S$ is a hull set of $G\overline{G}$ and $h(G\overline{G}) \leq t+2$, which completes the proof.

\end{proof}

Figure \ref{fig:prismas_e_diametros} shows examples of complementary prisms in which their minimum hull sets illustrate the upper bounds of Theorem \ref{theo:limiteSuperiorGdesconexo}. Black vertices represent a hull set of each complementary prism. The dashed circles represent the subgraph  $\overline{G}_1$ and dashed lines represent the set of edges joining every vertex of $\overline{G}_1$ to $v_1$ and $v_2$ in each complementary prism. 

\begin{figure}[htb]
\centering
{\setlength{\fboxsep}{12pt} 
\setlength{\fboxrule}{0.3pt} 
\fbox{
\subfigure[Complementary prism $G\overline{G}$ in which $diam(G_1) = 3$.]{
\psfrag{G}{$G$} \psfrag{Gb}{$\overline{G}$}
\psfrag{v1}{$v_1$} \psfrag{vb1}{$\overline{v}_1$}
\psfrag{v2}{$v_2$} \psfrag{vb2}{$\overline{v}_2$}
\psfrag{G1}{$G_1$} \psfrag{Gb1}{$\overline{G}_1$}
\includegraphics[width=0.3\textwidth]{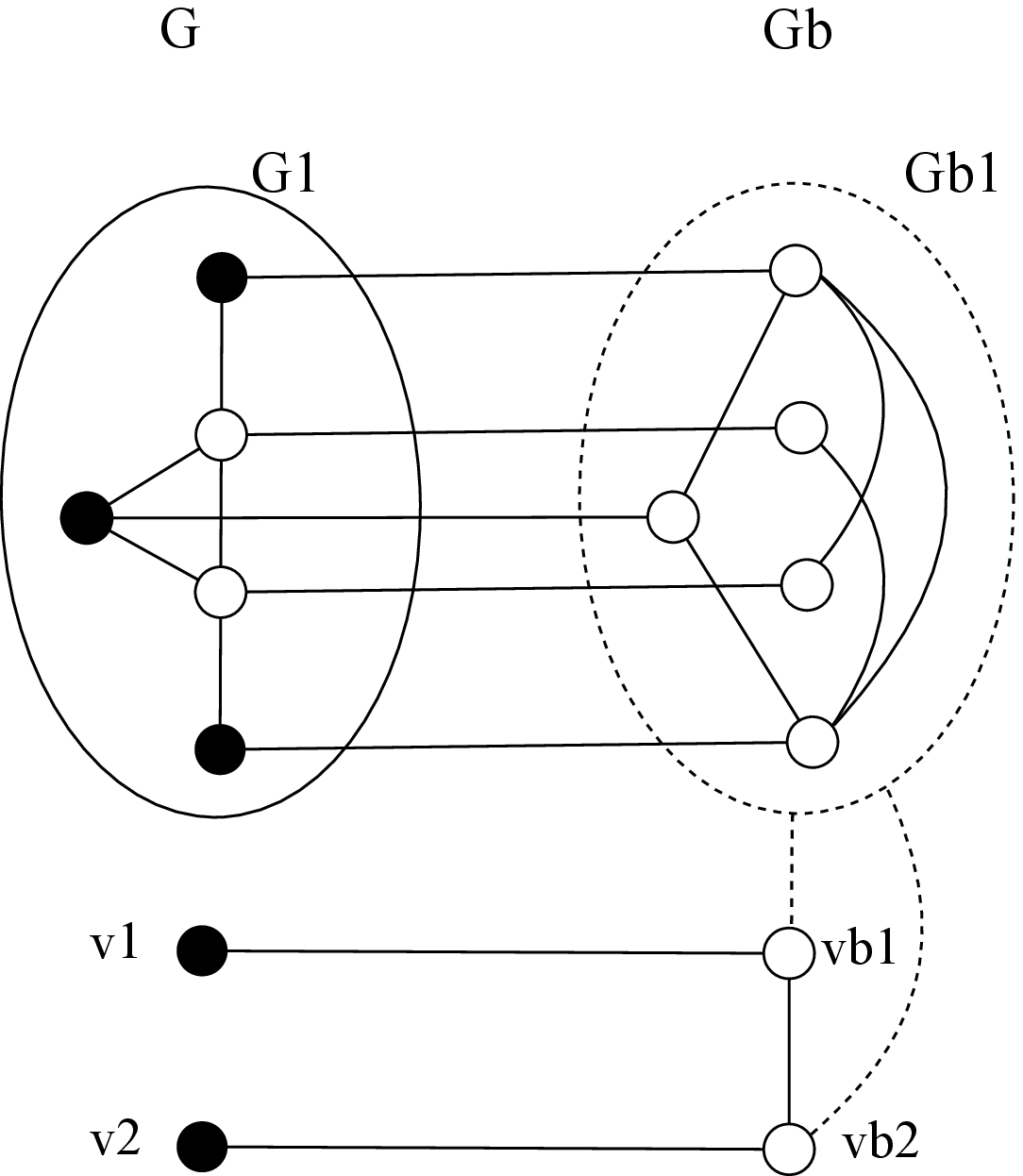}
\label{subfig:prisma_desconexo_uma_componente_nao_trivial}
} \qquad
\subfigure[Complementary prism $G\overline{G}$ in which $diam(G_1) = 4$.]{
\psfrag{G}{$G$} \psfrag{Gb}{$\overline{G}$}
\psfrag{G1}{$G_1$} \psfrag{Gb1}{$\overline{G}_1$}
\psfrag{u1}{$u_1$} \psfrag{ub1}{$\overline{u}_1$}
\psfrag{u2}{$u_2$} \psfrag{ub2}{$\overline{u}_2$}
\psfrag{u3}{$u_3$} \psfrag{ub3}{$\overline{u}_3$}
\psfrag{u4}{$u_4$} \psfrag{ub4}{$\overline{u}_4$}
\psfrag{u5}{$u_5$} \psfrag{ub5}{$\overline{u}_5$}
\psfrag{v1}{$v_1$} \psfrag{vb1}{$\overline{v}_1$}
\psfrag{v2}{$v_2$} \psfrag{vb2}{$\overline{v}_2$}
\includegraphics[width=0.40\textwidth]{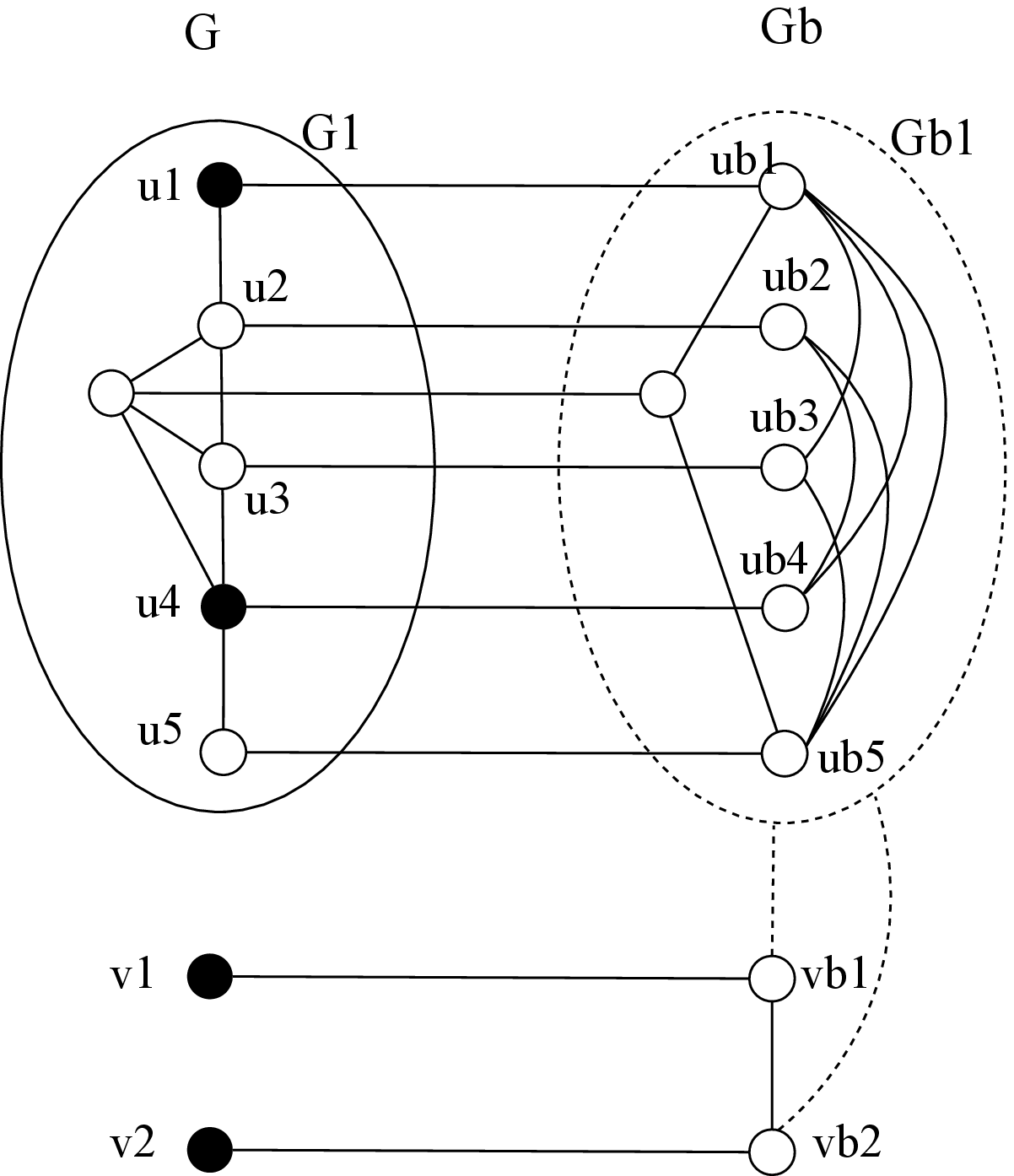}
\label{subfig:prisma_desconexo_diam_maior_que_tres}
}
}}
\caption{Complementary prisms $G\overline{G}$ of disconnected graphs $G$ with one nontrivial component.}
\label{fig:prismas_e_diametros}
\end{figure}

Theorem \ref{theo:limiteInferiorGdesconexo} and Theorem \ref{theo:limiteSuperiorGdesconexo} \eqref{theo:limiteSuperiorGdesconexoCasoContrario}, implies directly in the equality of Corollary  \ref{cor:t+2}.

\begin{corollary} \label{cor:t+2}
Let $G$ be a disconnected graph with one nontrivial connected component $G_1$ and $t > 0$ trivial components. If $diam(G_1) > 3$, then $h(G\overline{G}) = t + 2$.
\end{corollary}

Continuing our studies on complementary prisms $G\overline{G}$ when $G$ is a disconnected graph with one nontrivial connected component $G_1$, we notice that the upper bound of $h(G\overline{G})$ can also be determined in function of $h(\overline{G}_1)$. Theorem \ref{theo:limiteSuperiordiam2} shows this result.

\begin{theorem} \label{theo:limiteSuperiordiam2}
Let $G$ be a disconnected graph with exactly one nontrivial connected component. Consider $G_1$ the nontrivial component of $G$, $\overline{G}_1$ the corresponding component of $G_1$ in $\overline{G}$ and $t > 0$ the number of trivial components in $G$. If $diam(\overline{G}_1) \leq 2$, then $h(G\overline{G}) \leq h(\overline{G}_1) + t$.
\end{theorem}

\begin{proof}
Let $G$ be a disconnected graph with exactly one nontrivial connected component, denoted by $G_1$. Consider $v_1, \dots, v_t$ the vertices of the $t > 0$ trivial components of $G$ and $V(G_1) = \{u_1, \dots, u_p\}$.

Let $S_2$ a minimum hull set of $\overline{G}_1$. Suppose that $diam(\overline{G}_1) \leq 2$. We have that $d_{G\overline{G}}(\overline{u}_i,\overline{u}_j) \leq 2$, for every $\overline{u}_i,\overline{u}_j \in V(\overline{G}_1)$.
Consider $S \subseteq V(G\overline{G})$ such that $S = \{u_i \in V(G) : \overline{u}_i \in S_2 \} \cup \{v_1, \dots, v_t\}$.

Since $u \overline{u} \, \overline{v}_i v_i$ is a shortest path from $u$ to $v_i$, for every $u \in V(G_1)$ and $1 \leq i \leq t$, we have that $\overline{u}, \overline{v}_i \in I[S]$, for every $1 \leq i \leq t$ and $\overline{u} \in S_2$. Since $diam(\overline{G}_1) \leq 2$, we have that the distances of the vertices of $\overline{G}_1$ are the same in $G\overline{G}$, and since $S_2$ is a hull set of $\overline{G}_1$, then $H(S)$ contains $V(\overline{G}_1)$. Consequently $V(\overline{G}) \subseteq H(S)$.

Since $G_1$ is connected and $S \cap V(G_1) \neq \emptyset$, every uncontaminated vertex $u \in V(G)$ is contaminated by the vertices from a shortest path between $\overline{u}$ and its contaminated predecessor in $G_1$, then $H(S) = V(G\overline{G})$. Therefore $S$ is a hull set of $G\overline{G}$ and $h(G\overline{G}) \leq h(\overline{G}_1) + t$.

\end{proof}

Finally, considering that $diam(G_1) \leq 3$ and $diam(\overline{G}_1) \leq 2$, combining the first inequality of Theorem \ref{theo:limiteSuperiorGdesconexo} \eqref{theo:limiteSuperiorGdesconexoMenorIgual3} with the result stated in Theorem \ref{theo:limiteSuperiordiam2}, we can obtain a more restricted upper bound. Corollary \ref{cor:diam3e2} shows this result.

\begin{corollary} \label{cor:diam3e2}
Let $G$ be a disconnected graph with exactly one nontrivial connected component. Consider $G_1$ the nontrivial component of $G$, $\overline{G}_1$ the corresponding component of $G_1$ in $\overline{G}$ and $t > 0$ the number of trivial components in $G$. If $diam(G_1) \leq 3$ and $diam(\overline{G}_1) \leq 2$, then $h(G\overline{G}) \leq min\{h(G_1),h(\overline{G}_1)\} + t$.
\end{corollary}

\subsection{Cographs}


A \textit{cograph} is a graph with no induced $P_4$. As discussed in Introduction, Dourado et al. \cite{dourado2009computation} prove that deciding whether $h(G) \leq k$ is NP-complete. They also present polynomial-time algorithms for computing $h(G)$ when $G$ is a unit interval graph, a cograph or a split graph. We investigate the hull number of the complementary prisms of cographs, which resulted in Theorem \ref{theo:hullCograph}. 

Since a nontrivial cograph $G$ is connected if and only if $\overline{G}$ is disconnected \cite{corneil1981complement, seinsche1984on}, we can use our results of complementary prisms of disconnected graphs to prove results for complementary prisms of cographs.

\begin{theorem} \label{theo:hullCograph}
Let $G$ be a connected cograph. Let $k$ be the number of nontrivial components of $\overline{G}$ and let $t$ be the number of trivial components of $\overline{G}$. Then:
\begin{enumerate}[label=(\roman*)]
\item $h(G\overline{G}) = t$, if $k = 0$; \label{theo:hullCographCaseA}
\item $t + 2 \leq h(G\overline{G}) \leq min\{h(G_1),h(\overline{G}_1)\} + t$, if $k = 1$, in which $\overline{G}_1$ is the nontrivial component of $\overline{G}$ and $G_1$ is the corresponding component of $\overline{G}_1$, in $G$; \label{theo:hullCographCaseB}
\item $h(G\overline{G}) = k + t + 1$, if $k \geq 2$. \label{theo:hullCographCaseC}
\end{enumerate}
\end{theorem}

\begin{proof} 
\ref{theo:hullCographCaseA} If $k = 0$, $G$ is a complete graph on $t$ vertices and $G\overline{G} = K_t\overline{K}_t$, therefore $h(G\overline{G}) = t$ \cite{duarte2015prismas}.

\ref{theo:hullCographCaseB} Let $k = 1$. If $t = 0$, we have that $\overline{G}$ is connected, contradicting our assumption that $G$ is connected. Thus, we consider that $t > 0$. Since $G\overline{G}$ is isomorphic to $\overline{G}G$ and $\overline{G}$ is disconnected with one nontrivial component, Theorem \ref{theo:limiteInferiorGdesconexo} implies that $h(G\overline{G}) \geq t + 2$.
 
Since $G$ is a cograph, then $G_1$ is also a cograph. This implies that $\overline{G}_1$ is also a cograph and consequently $diam(\overline{G}_1) \leq 2$. Since $\overline{G}_1$ is a connected cograph, then $G_1$ is disconnected. This way, we consider the diameter of $G_1$ in terms of its $j$ connected components, denoted by $G_{1}^{i}$, $1 \leq i \leq j$.

Since $G_1$ is a cograph, every subgraph $G_{1}^{i}$, for every $1 \leq i \leq j$, are cographs, which implies that $diam(G_{1}^{i}) \leq 2$.  Since the hull number of a disconnected graph is equal to the sum of the hull numbers of its components, we have that $h(G_1) = \sum_{i = 1}^j h(G_{1}^{i})$.

Since $diam(G_1^i) \leq 2$, for every $1 \leq i \leq j$, and $diam(\overline{G}_1) \leq 2$, Corollary \ref{cor:diam3e2} implies that $h(G\overline{G}) \leq min\{\sum_{i = 1}^j h(G_{1}^{i}),h(\overline{G}_1)\} + t$, that is $h(G\overline{G}) \leq min\{h(G_1),h(\overline{G}_1)\} + t$. Therefore, if $k = 1$, then $t + 2 \leq h(G\overline{G}) \leq min\{h(G_1),h(\overline{G}_1)\} + t$. 

\ref{theo:hullCographCaseC} Let $k \geq 2$. Since $\overline{G}$ has at least two nontrivial connected components and $G\overline{G}$ is isomorphic to $\overline{G}G$, Theorem \ref{theo:Gdesconexo} implies that $h(G\overline{G}) = k + t + 1$.

\end{proof}

\subsection{Unlimited Geodetic Hull Number}

Unlike $P_3$-convexity, in which the hull number of complementary prisms $G\overline{G}$ when $G$ and $\overline{G}$ are connected is limited to $5$ \cite{duarte2015complexity}, in geodetic convexity the hull number of complementary prisms $G\overline{G}$ when $G$ and $\overline{G}$ are both connected cannot be limited. Theorem \ref{theo:Gconexo} express this result.

\begin{theorem}
\label{theo:Gconexo}
For every integer $n \geq 2$, there exist connected graphs $G$ and $\overline{G}$ such that $h(G\overline{G}) = n$. 
\end{theorem}

\begin{proof}
Suppose that $n = 2$. Let $G = P_4$ and $\overline{G} = \overline{P}_4$. By Theorem \ref{theo:duarte}  $h(P_4\overline{P}_4) = 2$. Therefore, for $n = 2$ the result holds.

Suppose that $n > 2$.
Let $K_n$ be a complete graph with $V(K_n) = \{u_1,u_2,...,u_n \}$. Let $G$ be a connected graph in which $V(G) = V(K_n) \cup \{v_1, v_2\}$ and $E(G) = E(K_n) \cup \{u_1v_1, u_2v_2\}$, see Figure \ref{fig:Gconexo}.

For the lower bound, let us prove that $h(G\overline{G}) \geq n$. We show that, for every $S \subseteq V(G\overline{G})$ such that $|S| \leq n-1$, $S$ is not a hull set $G\overline{G}$. For that, we verify that, if there exist two vertices $u_i,\overline{u}_i \in V(G\overline{G})$ such that $u_i,\overline{u}_i \notin S$, for some $3 \leq i \leq n$, then $u_i, \overline{u}_i$ will not belong to $H(S)$.

Let $S \subseteq V(G\overline{G})$ such that $|S| \leq n-1$. For every combination of $n-1$ vertices of $V(G\overline{G})$, at least a vertex $u_i$ and its corresponding vertex $\overline{u}_i$, for $1 \leq i \leq n$, will not belong to $S$, since $|V(K_n)| = |V(\overline{K}_n)| = n$. We consider some observations.

\bigskip

\noindent {\bf Observation 1:} {\it If $u_1, \overline{u}_1 \notin S$ ($u_2,\overline{u}_2 \notin S$), then there exist vertices $u_i,\overline{u}_i$ that do not belong to $S$, for $2 \leq i \leq n$ ($1 \leq i \leq n$, $i \neq 2$).}
\bigskip
 
\noindent Suppose that $u_1,\overline{u}_1 \notin S$.
For every combination $S$ with at most $n-1$ vertices of $V(G\overline{G})$, in which $u_1,\overline{u}_1$ do not belong to  $S$, we have that $u_1,\overline{u}_1$ can belong or not to $H(S)$. If $u_1,\overline{u}_1$ do not belong to $H(S)$, immediately $S$ is not a hull set of $G\overline{G}$. But if $u_1 \in H(S)$, we have that $u_1$ belongs to a shortest path between two other vertices that belong to $S$ or $I^\alpha[S]$. Since the distance between each pair of vertices $u_j, u_k \in V(K_n)$ is equal to $1$, then $u_1$ depends of $\overline{u}_1$ to belong to $H(S)$. 

For that $\overline{u}_1$ belongs to $H(S)$, $\overline{u}_1$ must belong to shortest path between two other vertices that belong to $S$ or $I^{\alpha-1}[S]$. Since $\overline{u}_1$ is adjacent only to $\{u_1, \overline{v}_2\}$, then for that $\overline{u}_1$ belongs to $I^\alpha[S]$, it is necessary that $u_1$ enters first in $I^{\alpha-1}[S]$ or $\overline{u}_1$ and $u_1$ enter together in $I^\alpha[S]$. Then if $u_1 \notin H(S)$, $\overline{u}_1 \notin H(S)$. 

Since  $\overline{u}_1 \notin S$, and since the distance from $\overline{v}_2$ to any neighbor of $u_1$ is equal to  $2$, then $u_1$ belongs to a shortest path between $v_1$ and $u_j$, for some $2 \leq j \leq n$, or between $v_1$ and $v_2$. 
Thus, if $u_1$ belongs to $H(S)$, then at least $v_1$ belongs to $I^\alpha[S]$. 
But since $v_1$ is adjacent only to $\{u_1, \overline{v}_1\}$, and the distance between $\overline{v}_1$ and any other vertex of $V(G\overline{G})$ is at most $2$, then $v_1$ must belong to $S$. This way, $S$ contains at most $n-2$ vertices of $V(K_n)\cup V(\overline{K}_n)$. Then at least a vertex $u_i$, for $2 \leq i \leq n$, and also its corresponding vertex $\overline{u}_i$ do not belong to $S$, since $|V(K_n)| = |V(\overline{K}_n)| = n$. 

The same argument can be applied to show that if $u_2,\overline{u}_2 \notin S$, then there exists vertices  $u_i,\overline{u}_i$ that do not belong to $S$, for $1 \leq i \leq n$, $i \neq 2$.

\bigskip

From the arguments of Observation 1 we conclude that, if $u_1,\overline{u}_1$ ($u_2,\overline{u}_2$) do not belong to $S$, it is possible that $u_1,\overline{u}_1$ ($u_2,\overline{u}_2$) belong to $H(S)$. If that happens, other vertices $u_i,\overline{u}_i$, for some $2 \leq i \leq n$, ($1 \leq i \leq n$, $i \neq 2$) must not belong to $S$.

\bigskip

\noindent {\bf Observation 3:} {\it If $u_1,u_2, \overline{u}_1,\overline{u}_2 \notin S$,  then there exist vertices $u_i,\overline{u}_i$ that do not belong to $S$, for $3 \leq i \leq n$.}

\bigskip

\noindent Consider that $u_1,\overline{u}_1,u_2,\overline{u}_2 \notin S$. We have that $u_1,u_2$ and their corresponding vertices $\overline{u}_1,\overline{u}_2$ can belong or not to $H(S)$. If those four vertices do not belong to $H(S)$, immediately $S$ is not a hull set of $V(G\overline{G})$. But if $u_1,u_2$ belong to $H(S)$, since the distance between each pair of vertices $u_j, u_k \in V(K_n)$ is equal to $1$, then $u_1$ does not belong to any shortest path between $u_j$ and $u_k$, and the same occurs with $u_2$.
Thus, we have that $u_1, u_2$ belong to the same shortest path between two other vertices that belong to $S$ or $I^\alpha[S]$. Since $\overline{u}_1,\overline{u}_2 \notin S$, then $u_1,u_2$ belong to a shortest path between $v_1$ and $v_2$, consequently $v_1,v_2 \in I^\alpha[S]$. But $v_1$ is adjacent only to $\{u_1,\overline{v}_1\}$ and $v_2$ is adjacent only to $\{u_2,\overline{v}_2\}$ and since $d_{G\overline{G}}(\overline{v}_1,\overline{v}_2) = 1$, then $v_1,v_2$ must belong to $S$. Thus, every combination of $S$, in this case, must contain $n-3$ vertices of $V(K_n)\cup V(\overline{K}_n)$, which implies that at least a vertex $u_i \in V(K_n)$, for some $3 \leq i \leq n$, as well as its corresponding vertex $\overline{u}_i \in V(\overline{K}_n)$, do not belong to $S$, since $|V(K_n)| = |V(\overline{K}_n)| = n$. 

\bigskip

So far we conclude that, if $|S| \leq n-1$, then some pair of vertices $u_i, \overline{u}_i$ must not belong to $S$. If $u_1, \overline{u}_1 \notin S$ or $u_2, \overline{u}_2 \notin S$ or $u_1,\overline{u}_1,u_2,\overline{u}_2 \notin S$, then $u_i, \overline{u}_i \notin S$, for $3 \leq i \leq n$.

Now, remains to show that, if $u_i,\overline{u}_i \notin S$, then $u_i,\overline{u}_i \notin H(S)$, for $3 \leq i \leq n$.

Since $u_i, \overline{u}_i$, for $3 \leq i \leq n$, are simplicial vertices in $G$ and $\overline{G}$, respectively, Lemma \ref{lema:simpliciaisPrisma} implies that every hull set $S$ of $G\overline{G}$ intersects $\{u_i, \overline{u}_i\}$. Thus, since $S \cap \{u_i, \overline{u}_i\} = \emptyset$, then $u_i,\overline{u}_i \notin H(S)$, a contradiction.

Therefore, every combination $S$, containing at most $n-1$ vertices of $V(G\overline{G})$, is not a hull set of $G\overline{G}$.
Consequently, $h(G\overline{G}) \geq n$.

For the upper bound, let $S = \{\overline{u}_1,\overline{u}_2,...,\overline{u}_n\}$. We have that $d_{G\overline{G}}(\overline{u}_1,\overline{u}_2) = 3$, then $u_1,u_2,\overline{v}_1, \overline{v}_2 \in I[S]$. Since $d_{G\overline{G}}(u_1, \overline{v}_1) = 2$ and $d_{G\overline{G}}(u_2, \overline{v}_2) = 2$, then $v_1, v_2 \in I^2[S]$. We also have that every $u_i$, for $3 \leq i \leq n$, is in a shortest path from $u_1$ to $\overline{u}_i$, then $u_i \in I^2[S]$. Since $V(G\overline{G}) = I^2[S] = H(S)$, then $S$ is a hull set of $G\overline{G}$. Therefore $h(G\overline{G}) = n$, which completes the proof.
\end{proof}

Figure \ref{fig:Gconexo} shows an example of the construction of the complementary prism $G\overline{G}$ presented in Theorem \ref{theo:Gconexo}, for $n = 6$. Black vertices represent a hull set of $G\overline{G}$.

\begin{figure}[htb]
\centering
{\setlength{\fboxsep}{9pt} 
\setlength{\fboxrule}{0.3pt} 
\fbox{
\psfrag{G}{$G$} \psfrag{Gb}{$\overline{G}$}
\psfrag{u1}{$u_1$} \psfrag{ub1}{$\overline{u}_1$}
\psfrag{u2}{$u_2$} \psfrag{ub2}{$\overline{u}_2$}
\psfrag{ui}{$u_3$} \psfrag{ubi}{$\overline{u}_3$}
\psfrag{uk}{$u_4$} \psfrag{ubk}{$\overline{u}_4$}
\psfrag{uj}{$u_5$} \psfrag{ubj}{$\overline{u}_5$}
\psfrag{un}{$u_6$} \psfrag{ubn}{$\overline{u}_6$}
\psfrag{v1}{$v_1$} \psfrag{vb1}{$\overline{v}_1$}
\psfrag{v2}{$v_2$} \psfrag{vb2}{$\overline{v}_2$}
\includegraphics[width=0.42\textwidth]{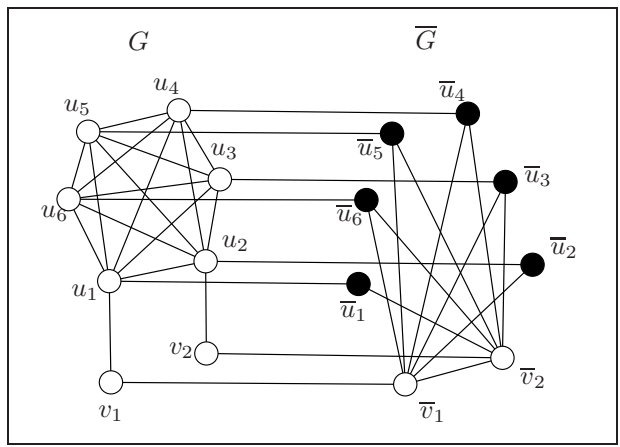}
}}
\caption{$G$ and $\overline{G}$ connected, with $h(G\overline{G}) = 6$.}
\label{fig:Gconexo}
\end{figure}

Finally, considering the complementary prism $G\overline{G}$ of the graphs or $G$ or $\overline{G}$ connected, the geodetic hull number also cannot be limited.  An example is the graph $S_n\overline{S}_n$ (Theorem \ref{theo:arvoresGeodesica}).

\section{Conclusions} 
\label{sec:conclusions}

We have considered  the geodetic hull number on complementary prisms of trees, disconnected graphs and cographs. We  have also shown that geodetic hull number on complementary prisms $G\overline{G}$ can be unlimited on connected graphs $G$ and $\overline{G}$, unlike what happens in $P_3$-convexity.

As future work, we suggest to determine the complexity of the following decision problem.

\begin{problem}
Let $k$ be a positive integer. Given a graph $G$, to decide whether the geodetic hull number of the complementary prism $G\overline{G}$ is at most $k$.
\end{problem}

\section*{Acknowledgement} 
\label{sec:acknowledgements}

The authors are partially supported by CAPES, CNPq, and FAPERJ.

%
%

\bibliographystyle{abbrv}
\bibliography{bibliografia_convexidade}

\begin{thebibliography}{10}

\bibitem{araujo2011hull}
J.~Araujo, V.~Campos, F.~Giroire, L.~Sampaio, and R.~Soares.
\newblock On the hull number of some graph classes.
\newblock {\em Electronic Notes in Discrete Mathematics}, 38:49--55, 2011.

\bibitem{balogh1998random}
J.~Balogh and G.~Pete.
\newblock Random disease on the square grid.
\newblock {\em Random Structures and Algorithms}, 13(3-4):409--422, 1998.

\bibitem{barbosa2012caratheodory}
R.~M. Barbosa, E.~M.~M. Coelho, M.~C. Dourado, D.~Rautenbach, and J.~L.
  Szwarcfiter.
\newblock On the {Carath{\'e}odory} number for the convexity of paths of order
  three.
\newblock {\em SIAM Journal on Discrete Mathematics}, 26(3):929--939, 2012.

\bibitem{bollobas2006art}
B.~Bollob{\'a}s.
\newblock {\em The Art of Mathematics: Coffee Time in Memphis}.
\newblock Cambridge University Press, 2006.

\bibitem{brunetti2012minimum}
S.~Brunetti, G.~Cordasco, L.~Gargano, E.~Lodi, and W.~Quattrociocchi.
\newblock Minimum weight dynamo and fast opinion spreading.
\newblock In {\em International Workshop on Graph-Theoretic Concepts in
  Computer Science}, pages 249--261. Springer, 2012.

\bibitem{buckley1990distance}
F.~Buckley and F.~Harary.
\newblock {\em Distance in Graphs}.
\newblock Addison-Wesley Longman, 1990.

\bibitem{caceres2010geodetic}
J.~C{\'a}ceres, C.~Hernando, M.~Mora, I.~M. Pelayo, and M.~L. Puertas.
\newblock On the geodetic and the hull numbers in strong product graphs.
\newblock {\em Computers \& Mathematics with Applications}, 60(11):3020--3031,
  2010.

\bibitem{caceres2006geodetic}
J.~C{\'a}ceres, C.~Hernando, M.~Mora, I.~M. Pelayo, M.~L. Puertas, and
  C.~Seara.
\newblock On geodetic sets formed by boundary vertices.
\newblock {\em Discrete Mathematics}, 306(2):188--198, 2006.

\bibitem{campos2015graphs}
V.~Campos, R.~M. Sampaio, A.~Silva, and J.~L. Szwarcfiter.
\newblock Graphs with few {$P_4$'s} under the convexity of paths of order
  three.
\newblock {\em Discrete Applied Mathematics}, 192:28--39, 2015.

\bibitem{canoy2005hull}
S.~R. Canoy~Jr and G.~B. Cagaanan.
\newblock On the hull number of the composition of graphs.
\newblock {\em Ars Combinatoria}, 75:113--120, 2005.

\bibitem{canoy2006convexity}
S.~R. Canoy~Jr, G.~B. Cagaanan, and S.~V. Gervacio.
\newblock Convexity, geodetic, and hull numbers of the join of graphs.
\newblock {\em Utilitas Mathematica}, 71:143--159, 2006.

\bibitem{centeno2011irreversible}
C.~C. Centeno, M.~C. Dourado, L.~D. Penso, D.~Rautenbach, and J.~L.
  Szwarcfiter.
\newblock Irreversible conversion of graphs.
\newblock {\em Theoretical Computer Science}, 412(29):3693--3700, 2011.

\bibitem{centeno2013geodetic}
C.~C. Centeno, L.~D. Penso, D.~Rautenbach, and V.~G.~P. de~S\'{a}.
\newblock Geodetic number versus hull number in ${P_3}$-convexity.
\newblock {\em SIAM Journal on Discrete Mathematics}, 27(2):717--731, 2013.

\bibitem{changat2001all}
M.~Changat, S.~Klav{\v{z}}ar, and H.~M. Mulder.
\newblock The all-paths transit function of a graph.
\newblock {\em Czechoslovak mathematical journal}, 51(2):439--448, 2001.

\bibitem{changat1999triangle}
M.~Changat and J.~Mathew.
\newblock On triangle path convexity in graphs.
\newblock {\em Discrete Mathematics}, 206(1):91--95, 1999.

\bibitem{chartrand2000hull}
G.~Chartrand, F.~Harary, and P.~Zhang.
\newblock On the hull number of a graph.
\newblock {\em Ars Combinatoria}, 57:129--138, 2000.

\bibitem{coelho2014caratheodory}
E.~M.~M. Coelho, M.~C. Dourado, D.~Rautenbach, and J.~L. Szwarcfiter.
\newblock The {Carath\'{e}odory} number of the ${P_3}$ convexity of chordal
  graphs.
\newblock {\em Discrete Applied Mathematics}, 172:104 -- 108, 2014.

\bibitem{corneil1981complement}
D.~Corneil, H.~Lerchs, and L.~Burlingham.
\newblock Complement reducible graphs.
\newblock {\em Discrete Applied Mathematics}, 3(3):163 -- 174, 1981.

\bibitem{costa2015inapproximability}
E.~R. Costa, M.~C. Dourado, and R.~M. Sampaio.
\newblock Inapproximability results related to monophonic convexity.
\newblock {\em Discrete Applied Mathematics}, 197:70--74, 2015.

\bibitem{domingos2001mining}
P.~Domingos and M.~Richardson.
\newblock Mining the network value of customers.
\newblock In {\em Proceedings of the Seventh ACM SIGKDD International
  Conference on Knowledge Discovery and Data Mining}, KDD '01, pages 57--66,
  New York, NY, USA, 2001. ACM.

\bibitem{dourado2016near}
M.~C. Dourado, V.~G.~P. de~S\'{a}, D.~Rautenbach, and J.~L. Szwarcfiter.
\newblock Near-linear-time algorithm for the geodetic radon number of grids.
\newblock {\em Discrete Applied Mathematics}, 210:277--283, 2016.

\bibitem{dourado2009computation}
M.~C. Dourado, J.~G. Gimbel, J.~Kratochv{\'\i}l, F.~Protti, and J.~L.
  Szwarcfiter.
\newblock On the computation of the hull number of a graph.
\newblock {\em Discrete Mathematics}, 309(18):5668--5674, 2009.

\bibitem{dourado2016pkfree}
M.~C. Dourado, L.~D. Penso, and D.~Rautenbach.
\newblock On the geodetic hull number of ${P_k}$-free graphs.
\newblock {\em Theoretical Computer Science}, 640:52--60, 2016.

\bibitem{dourado2010hull}
M.~C. Dourado, F.~Protti, D.~Rautenbach, and J.~L. Szwarcfiter.
\newblock On the hull number of triangle-free graphs.
\newblock {\em SIAM Journal on Discrete Mathematics}, 23(4):2163--2172, 2010.

\bibitem{dourado2010some}
M.~C. Dourado, F.~Protti, D.~Rautenbach, and J.~L. Szwarcfiter.
\newblock Some remarks on the geodetic number of a graph.
\newblock {\em Discrete Mathematics}, 310(4):832--837, 2010.

\bibitem{dourado2010complexity}
M.~C. Dourado, F.~Protti, and J.~L. Szwarcfiter.
\newblock Complexity results related to monophonic convexity.
\newblock {\em Discrete Applied Mathematics}, 158(12):1268--1274, 2010.

\bibitem{dourado2013caratheodory}
M.~C. Dourado, D.~Rautenbach, V.~F. Dos~Santos, P.~M. Sch{\"a}fer, and J.~L.
  Szwarcfiter.
\newblock On the {Carath{\'e}odory} number of interval and graph convexities.
\newblock {\em Theoretical Computer Science}, 510:127--135, 2013.

\bibitem{dourado2016complexity}
M.~C. Dourado and R.~M. Sampaio.
\newblock Complexity aspects of the triangle path convexity.
\newblock {\em Discrete Applied Mathematics}, 206:39--47, 2016.

\bibitem{dreyer2009irreversible}
P.~A. Dreyer and F.~S. Roberts.
\newblock Irreversible k-threshold processes: Graph-theoretical threshold
  models of the spread of disease and of opinion.
\newblock {\em Discrete Applied Mathematics}, 157(7):1615--1627, 2009.

\bibitem{duarte2015prismas}
M.~A. Duarte.
\newblock {\em Sobre convexidade em prismas complementares}.
\newblock PhD thesis, Universidade Federal de Goi{\'a}s, 2015.

\bibitem{duarte2015complexity}
M.~A. Duarte, L.~Penso, D.~Rautenbach, and U.~dos Santos~Souza.
\newblock Complexity properties of complementary prisms.
\newblock {\em Journal of Combinatorial Optimization}, pages 1--8, 2015.

\bibitem{duchet1988convex}
P.~Duchet.
\newblock Convex sets in graphs, {II}. minimal path convexity.
\newblock {\em Journal of Combinatorial Theory, Series B}, 44(3):307--316,
  1988.

\bibitem{everett1985hull}
M.~G. Everett and S.~B. Seidman.
\newblock The hull number of a graph.
\newblock {\em Discrete Mathematics}, 57(3):217--223, 1985.

\bibitem{farber1986convexity}
M.~Farber and R.~E. Jamison.
\newblock Convexity in graphs and hypergraphs.
\newblock {\em SIAM Journal on Algebraic Discrete Methods}, 7(3):433--444,
  1986.

\bibitem{flocchini2004dynamic}
P.~Flocchini, E.~Lodi, F.~Luccio, L.~Pagli, and N.~Santoro.
\newblock Dynamic monopolies in tori.
\newblock {\em Discrete applied mathematics}, 137(2):197--212, 2004.

\bibitem{hammack2011handbook}
R.~Hammack, W.~Imrich, and S.~Klav{\v{z}}ar.
\newblock {\em Handbook of Product Graphs}.
\newblock CRC press, 2011.

\bibitem{harary1981convexity}
F.~Harary and J.~Nieminen.
\newblock Convexity in graphs.
\newblock {\em J. Differential Geom.}, 16(2):185--190, 1981.

\bibitem{hassin2001distributed}
Y.~Hassin and D.~Peleg.
\newblock Distributed probabilistic polling and applications to proportionate
  agreement.
\newblock {\em Information and Computation}, 171(2):248--268, 2001.

\bibitem{haynes2007complementary}
T.~W. Haynes, M.~A. Henning, P.~J. Slater, and L.~C. van~der Merwe.
\newblock The complementary product of two graphs.
\newblock {\em Bulletin of the Institute of Combinatorics and its
  Applications}, 51:21--30, 2007.

\bibitem{hernando2005steiner}
C.~Hernando, T.~Jiang, M.~Mora, I.~M. Pelayo, and C.~Seara.
\newblock On the {Steiner}, geodetic and hull numbers of graphs.
\newblock {\em Discrete Mathematics}, 293(1):139--154, 2005.

\bibitem{kempe2003maximizing}
D.~Kempe, J.~Kleinberg, and E.~Tardos.
\newblock Maximizing the spread of influence through a social network.
\newblock In {\em Proceedings of the Ninth ACM SIGKDD International Conference
  on Knowledge Discovery and Data Mining}, KDD '03, pages 137--146, New York,
  NY, USA, 2003. ACM.

\bibitem{kempe2005influential}
D.~Kempe, J.~Kleinberg, and {\'E}.~Tardos.
\newblock Influential nodes in a diffusion model for social networks.
\newblock In {\em International Colloquium on Automata, Languages, and
  Programming}, pages 1127--1138. Springer, 2005.

\bibitem{khoshkhah2014dynamic}
K.~Khoshkhah, H.~Soltani, and M.~Zaker.
\newblock Dynamic monopolies in directed graphs: The spread of unilateral
  influence in social networks.
\newblock {\em Discrete Applied Mathematics}, 171:81--89, 2014.

\bibitem{mustafa2004listen}
N.~H. Mustafa and A.~Pekec.
\newblock Listen to your neighbors: How (not) to reach a consensus.
\newblock {\em SIAM Journal on Discrete Mathematics}, 17(4):634--660, 2004.

\bibitem{nascimento2016complexity}
J.~R. Nascimento, E.~M.~M. Coelho, H.~Coelho, and J.~L. Szwarcfiter.
\newblock On the complexity of the {$P_3$}-hull number of the cartesian product
  of graphs.
\newblock {\em Electronic Notes in Discrete Mathematics}, 55:169--172, 2016.

\bibitem{peleg2002local}
D.~Peleg.
\newblock Local majorities, coalitions and monopolies in graphs: a review.
\newblock {\em Theoretical Computer Science}, 282(2):231--257, 2002.

\bibitem{penso2015complexity}
L.~D. Penso, F.~Protti, D.~Rautenbach, and U.~dos Santos~Souza.
\newblock Complexity analysis of ${P_3}$-convexity problems on bounded-degree
  and planar graphs.
\newblock {\em Theoretical Computer Science}, 607:83--95, 2015.

\bibitem{seinsche1984on}
D.~Seinsche.
\newblock On a property of the class of n-colorable graphs.
\newblock {\em Journal of Combinatorial Theory, Series B}, 16(2):191--193,
  1974.

\end{thebibliography}

\end{document}